\documentclass{article}

\RequirePackage[OT1]{fontenc} \RequirePackage{amsthm,amsmath} \RequirePackage[authoryear]{natbib}

\usepackage{amsfonts, amssymb, graphicx, mathtools, epstopdf} 
\usepackage{accents}

\usepackage{setspace}

\newtheorem{theorem}{Theorem}
\newtheorem{assumption}{Assumption}

\newtheorem{corollary}{Corollary}
\newtheorem{lemma}{Lemma}

\newcommand{\st}{{\textrm{subject to}}}
\newcommand{\Escr}{{\mathcal{E}}}
\newcommand{\CI}{{\operatorname{CI}}}
\newcommand{\obf}{{\bf 1}}

\title{\Large Using Exposure Mappings as Side Information in Experiments with Interference}
\author{\normalsize David Choi}

\begin{document}
\maketitle

\begin{abstract}
	Exposure mappings are widely used to model potential outcomes in the presence of interference, where each unit's outcome may depend not only on its own treatment, but also on the treatment of other units as well. However, in practice these models may be only a crude proxy for social dynamics. In this work, we give estimands and estimators that are robust to the misspecification of an exposure model. In the first part, we require the treatment effect to be nonnegative (or ``monotone'') in both direct effects and spillovers. In the second part, we consider a weaker estimand (``contrasts attributable to treatment'') which makes no restrictions on the interference at all.
\vskip.5cm
\noindent {\bf Keywords}: causal inference, interference, network data, exposure model, randomized experiment, peer effects
\end{abstract}

\doublespacing 

\section{Introduction}

With increasing frequency, randomized experiments are being proposed in which the object of study is an interconnected social network or societal system. Examples can be found in disparate domains such as health \citep{hudgens2008toward, miguel2004worms}, politics \citep{bond201261, coppock2014information}, crime \citep{verbitsky2012causal}, developmental economics \citep{banerjee2013diffusion}, and consumer demand \citep{bakshy2012social}. In each of these settings, it is believed that interdependencies between the units may play an important factor in determining individual outcomes, and conversely that local actions may interact and give rise to global phenomena, such as herd immunity \citep{ogburn2017vaccines} or cascading behavior \citep{leskovec2007patterns}. 

The analysis of such experiments poses statistical challenges, particularly when only a single instantiation of the network is available for study, so that the randomization of treatment is over the interconnected units within the network. In this case, the outcome of one unit is likely to be affected by the treatments and outcomes of others. This violates traditional methods for causal inference, which require an assumption of ``no interference between units''. While it is possible to relax this assumption, to do so one must model the underlying dependencies; typically this entails placing bounds on who can influence whom, how such effects might combine, and whether they can cascade over long distances. In some settings this may be an unreasonable modeling burden, with misspecification resulting in possible loss of validity and anti-conservative estimates.

In this paper, we propose new methods for experiments in network settings. These methods will resemble existing ones which use exposure mappings, a popular class of models for network experiments \citep{aronow2017estimating, eckles2017design, manski2013identification, vanderweele2012mapping}. However, unlike most previous approaches, the new methods will produce valid confidence regions (though possibly conservative) even when the misspecification is severe. In most of the paper, this will be accomplished by assuming that the treatment effect is nonnegative \citep{choi2017estimation}. For example, in a vaccination study it might not be reasonable to model the social interactions that would occur under every counterfactual, but it might be more reasonable to assume that withholding the vaccination treatment from a municipality would not improve outcomes in that municipality nor elsewhere through spillovers. At the end of the paper, we will remove the nonnegativity requirement, and consider an inference task that does not require structural assumptions on the underlying social mechanism.

Recent methodological works include experiment design \citep{jagadeesan2017designs}; testing \citep{athey2017exact, pouget2017testing, basse2017exact}; right-censored failure times \citep{loh2018randomization}; doubly robust methods for experiments and observational studies \citep{ogburn2017causal, sofrygin2017semi}; and markov random fields \citep{tchetgen2017auto}. Additionally, see review papers \citep{halloran2016dependent} and  \citep{aral2016networked} (particularly for the table of experiments in the latter). In all of these works, a correctly specified exposure mapping is required for estimation (though not for testing). Other recent works, including \citep{chin2018central, savje2017average}, seek to relax this assumption.

The organization of the paper is as follows: Section \ref{sec: formulation} describes a motivating example and gives the problem formulation. Section \ref{sec: method} describes a method assuming nonnegative treatment effects, with a simulation study in Section \ref{sec: sims} and data analysis example in Section \ref{sec: real data}. Section \ref{sec: nonmonotone} weakens the estimand and removes the assumption of nonnegative treatment effects. Proofs are contained in the appendices.

\section{Motivating Example and Problem Formulation} \label{sec: formulation}

\subsection{Motivating Example}

As a concrete example, we begin by describing an experiment studied in \citep{miguel2004worms}, whose data will be analyzed in Section \ref{sec: real data} as a demonstration of the proposed method. The experiment was a deworming project carried out in 1998 in Busia, Kenya, in order to reduce the number of infections by parasitic worms in young children. Schools in group 1 received free deworming treatments beginning in 1998, while group 2 did not. Students were surveyed one year later, and substantially fewer infections were found in the treatment-eligible pupils, with 141 infections in group 1 and 506 infections in group 2. However, it is believed that the number of infections in each school was affected not only by its own treatment status, but also that of other nearby schools as well. This is because students that received the deworming treatment were susceptible to re-infection by infected students. 

How might spatial information be beneficial in such a setting? Figure \ref{fig: cartoon} shows a stylized cartoon in which 9 schools are arranged on a 1-dimensional line, with slight grouping into 3 clusters (schools 1-3, 4-6, and 7-9). There seems to be interference due to re-infection; for example, school 5 is treated but next to two untreated schools in its same cluster, and its infection counts are similar to those of the untreated schools. 

\begin{figure}[h!]
	\begin{center}
		\includegraphics[width=4in]{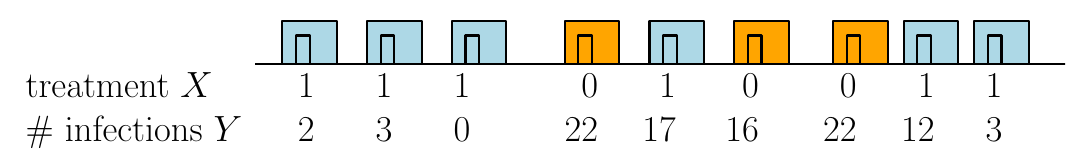}
	\end{center}
	\caption{Stylized cartoon of a hypothetical deworming experiment similar to  \citep{miguel2004worms}, showing 9 schools arranged on a 1-dimensional line. The infection counts are consistent with interference between nearby untreated and treated schools}
		\label{fig: cartoon}
\end{figure}

The hope is that by using this spatial information, it may be possible to ``disqualify'' some treated schools (such as school 5) for being too close to untreated ones, without compromising the randomization inference. To accomplish this, we might assume an exposure mapping. For each of the schools $i=1,\ldots,9$, let $Z_i$ denote the variable 
\[ Z_i = \begin{cases}1 & \text{$i$ treated, and $i\!-\!1$ and $i\!+\!1$ are treated or are not in $i$'s cluster} \\ 0 & \text{otherwise} \end{cases} \]
 For the experiment shown in Figure \ref{fig: cartoon}, it can be seen that restricting to the schools for which $Z_i=1$ (schools 1, 2, 3, and 9) greatly lowers the infection count, at the cost of a reduced sample size. Confidence intervals could then be computed using \citep{aronow2017estimating}, under the crucial assumption that $Z_i$ is correctly specified\footnote{In addition, we would also require $Z_i$ to be specified before the treatment and outcomes are observed, to avoid ``data snooping''.}; in particular, we would require that the potential outcomes of each school $i$ can be written as a function involving only $Z_i$. This could be a strong assumption; for example, our specification for $Z_i$ does not allow for the possibility that if school $i-2$ is untreated while schools $i-1$ and $i$ are both treated, then school $i-2$ might reinfect school $i-1$, which might then re-infect school $i$, and so forth leading to a cascade. Recent work such as \citep{savje2017average, chin2018central} seek to relax this; these require the interference to be asymptotically small (or ``sparse''), but still disallow the abovementioned cascade possibility. 

In the following sections, we will formalize the problem statement, and propose a new method that is able to use a hypothesized exposure mapping, without assuming correct specification of the generative model.

\subsection{Problem Formulation} 

Motivated by the deworming experiment of \citep{miguel2004worms}, we consider a randomized experiment on $N$ schools. Let $X = (X_1,\ldots,X_N)$ denote the treatment assignments, where $X_i=1$ denotes treatment and $X_i=0$ denotes control for the $i$th unit. Let $Y = (Y_1,\ldots,Y_N)$ denote the observed outcome -- specifically, let $Y_i$ denote the number of observed infections at school $i$. 

We will assume that units are assigned independently to treatment, with probability $\rho$:
\begin{equation}\label{eq: X}
X_i \stackrel{iid}{\sim} \operatorname{Bernoulli}(\rho),\qquad i=1,\ldots,N.
\end{equation}
We do not assume SUTVA, but instead allow for interference between units, and let each $Y_i$ be an idiosyncratic function of all $N$ treatment assignments so that
\begin{equation}\label{eq: Y}
Y_i = f_i(X_1,\ldots,X_N),\qquad i=1,\ldots,N,
\end{equation}
for some collection of functions $f_1,\ldots,f_N$. Let $\theta = (\theta_1,\ldots,\theta_N)$ denote the counterfactual outcome under full treatment, 
\begin{equation}\label{eq: theta}
\theta_i = f_i(1,\ldots,1).
\end{equation}
Our inferential goal (except in Section \ref{sec: nonmonotone}) will be to construct a valid one-sided confidence interval to upper bound $\bar{\theta}$, the number of infections that would have occurred under full treatment:
\begin{equation} \label{eq: bar theta}
\bar{\theta} = \frac{1}{N} \sum_{i=1}^N \theta_i,
\end{equation}
under the following assumption:
\begin{assumption}\label{as: monotone} It holds that
	\begin{equation} \label{eq: monotone}
 0 \leq \theta_i \leq Y_i,\qquad i=1,\ldots,N.
 \end{equation}
\end{assumption}
 Assumption \ref{as: monotone} is a structural assumption on the generative model. A sufficient condition for \eqref{eq: monotone} to hold is for full treatment to achieve the lowest possible number of infections at each school:
\[ \theta_i = \min_{x_{1:N}} f_i(x_1,\ldots,x_N), \qquad i=1,\ldots,N.\]
While untestable, this assumption might be viewed as a reasonable one in some applications; for example, in the deworming study, Assumption \ref{as: monotone} corresponds to the assumption that declining to treat an entire school would not result in fewer infections than full treatment. Additionally, even when such an assumption might be contentious, it might still be preferable (or at least comparable) to existing inferential approaches, such as assuming partial interference or an exposure model. This may be the case when the observed network/spatial information/exposure model is thought to only be a crude proxy for the underlying social mechanisms. 

Assumption \ref{as: monotone} was used previously in \citep{choi2017estimation} to analyze the same experiment of \citep{miguel2004worms}. In that analysis, however, no spatial or network information was used. To use such information, here we will propose a method which combines Assumption \ref{as: monotone} with a potentially misspecified exposure model.

\paragraph{Network/Spatial Information} For each unit $i \in [N]$ let $\eta_i \subset [N]$ the denote the set of units comprised of $i$ plus $i$'s ``neighborhood'', which are those units believed to have greatest influence upon $i$. For example, $\eta_i$ might include those units geographically closest to $i$, or it could be based on previously observed interactions between $i$ and the other units, or some other source of prior knowledge. We emphasize that we will \emph{not} use $\eta_i$ to make additional formal assumptions on the generative process. As a result, if Assumption \ref{as: monotone} holds but $\{\eta_i\}$ is a poor proxy for the underlying generative model, our confidence interval may lose power (resulting in wider intervals), but not validity or coverage. 

\section{Methodology} \label{sec: method}

The organization of this section is the following. 
\begin{enumerate}
	\item In Section \ref{sec: notation}, we introduce the quantity $\CI^{\text{ideal}}$ given by
	\[ \CI^{\text{ideal}} = \hat{\theta} + z_{1-\alpha}\frac{\sqrt{\hat{\sigma}}}{L},\]
	where $L$ is the number of ``treated neighborhoods'', $\hat{\theta}$ is a point estimate for $\bar{\theta}$, and $\hat{\sigma}$ estimates the variance of $L(\hat{\theta} - \bar{\theta})$; see \eqref{eq: L},  \eqref{eq: theta hat}, and \eqref{eq: hat sigma} for precise expressions. This quantity will upper bound $\bar{\theta}$ while utilizing the neighborhoods $\{\eta_i\}$, but will require knowledge of the unobserved counterfactual vector $\theta$. As a result $\CI^{\text{ideal}}$ cannot be evaluated, but will be the basis for the eventual method.
	\item In Section \ref{sec: CLT}, we establish that $\CI^{\text{ideal}}$ is asymptotically a valid $(1-\alpha)$ confidence upper bound on $\bar{\theta}$.	
\item In Section \ref{sec: upper bound}, we show that while $\CI^{\text{ideal}}$ cannot be evaluated (since $\theta$ is unobserved), in some settings it can be upper bounded using the observations $Y$. Specifically, we introduce quantities $\hat{\theta}_Y$ and $\tilde{\sigma}_Y$ which are proxies for $\hat{\theta}$ and $\hat{\sigma}$ using $Y$ instead of the unobserved $\theta$, and show that if
\[ 1 - z_{1-\alpha}\frac{\hat{\theta}_Y}{\sqrt{\tilde{\sigma}_Y}}\cdot \frac{2Np(1-p)}{L} > 0,\]
where $p$ is the probability of treatment for a neighborhood, then the quantity $\CI^{\text{obs}}$ given by
\[ \CI^{\text{obs}} = \hat{\theta}_Y + z_{1-\alpha} \frac{\sqrt{\tilde{\sigma}_Y}}{L}\]
is an upper bound for $\CI^{\text{ideal}}$. As a result, if $\CI^{\text{ideal}}$ is asymptotically valid, then $\CI^{\text{obs}}$ is also an asymptotic $(1-\alpha)$ confidence upper bound for $\bar{\theta}$.
\item Section \ref{sec: variants} presents further discussion.
\end{enumerate}

\subsection{Idealized Network-based Confidence Interval} \label{sec: notation}

For $i=1,\ldots,N$, let $h_i$ denote an exposure mapping which takes the treatment variables $\{X_j: j \in \eta_i\}$ corresponding to $i$'s neighborhood, and returns $0$ or $1$. Two examples are
\begin{enumerate}
	\item $h_i$ returns $1$ if all units in $\eta_i$ are treated:
	\[h_i(\{X_j:j\in \eta_i\}) = \prod_{j \in \eta_i} X_j, \qquad i=1,\ldots,N.\]
	\item $h_i$ returns $1$ if $i$ is treated and at least $d_{\min}$ units in $\eta_i$ are treated:
	\begin{equation} \label{eq: sim h}
	h_i(\{X_j: j \in \eta_i\}) = X_i \cdot 1\left\{ \sum_{j \in \eta_j} X_j \geq d_{\min}\right\}, \qquad i=1,\ldots,N. 
	\end{equation}
\end{enumerate}	
We emphasize that $\{h_i\}$ will {\it not} place strong assumptions on the generative process; our confidence bound on $\bar{\theta}$ will be valid for any choice of $\{h_i\}$, although it will require Assumption \ref{as: monotone}.

Given the observed treatment assignment vector $X$, let $Z = (Z_1,\ldots,Z_N)$ denote the effective treatment of unit $i$:
\begin{equation}\label{eq: Z}
Z_i = h_i(\{X_j: j \in \eta_i\}), \qquad i=1,\ldots,N.
 \end{equation}
We will require that each $Z_i$ has the same probability of equaling 1:
\[\mathbb{P}(Z_i=1)=p, \qquad i=1,\ldots,N.\]
In practice, this will usually mean that the neighborhoods $\eta_i$ are the same size. (In Section \ref{sec: variants}, we discuss why relaxing this constraint is apparently non-trivial). 
%

Let $L$ denote the number of effective treatments,
\begin{equation}\label{eq: L}
L = \sum_{i=1}^N Z_i.
\end{equation}
Let $\hat{\theta}$ denote the average of $\theta$ over the effective treatments,
\begin{equation} \label{eq: theta hat}
\hat{\theta} = \frac{1}{L} \sum_{i=1}^N \theta_i Z_i.
\end{equation}
Let $T$ denote the statistic
\begin{equation}\label{eq: T}
T(Z;\theta) = \sum_{i=1}^N (\bar{\theta} - \theta_i)Z_i,
\end{equation}
which can be seen to equal $L(\bar{\theta} - \hat{\theta})$. The variance of $T$ can be seen to equal
\begin{equation}\label{eq: var T}
\operatorname{Var}(T) = (\theta - \bar{\theta})^T P (\theta - \bar{\theta}),
\end{equation}
where $P \in \mathbb{R}^{N \times N}$ is the matrix of pairwise joint probabilities,
\[ P_{ij} = \mathbb{P}(Z_iZ_j = 1), \qquad i,j = 1,\ldots,N.\]
We will require $P_{ij}$ to be bounded away from zero,
\[\mathbb{P}(Z_i=1, Z_j=1) \geq p_{\min},\qquad i,j=1,\ldots,N.\]
 

We will consider the following estimate of $\operatorname{Var}(T)$, whose consistency will be established in Section \ref{sec: CLT}:
\begin{equation}\label{eq: hat sigma}
\hat{\sigma} = N p(1-p)\left(\frac{1}{L}\sum_{i=1}^N (\theta_i - \hat{\theta})^2 Z_i\right)  + \sum_{i=1}^N \sum_{j=1}^N \theta_i \theta_j \frac{\Escr_{ij}}{P_{ij}} Z_i Z_j,
\end{equation}
where the matrix $\Escr \in \mathbb{R}^{N \times N}$ and its uncentered version $R \in \mathbb{R}^{N \times N}$ are given by
\begin{align}
	R & = P - p(1-p)I - p^2 \obf \obf^T  \label{eq: R} \\
	\Escr & =  \left(I - \frac{ \obf \obf^T}{N}\right)^T R \left(I - \frac{ \obf \obf^T}{N}\right),  \label{eq: Escr}
\end{align}
where $\obf$ is the vector of all ones. It can be seen that $R$ is the difference between $P$ and the second moment matrix of $N$ independent $\operatorname{Bernoulli}(p)$ random variables, and that $\Escr$ is a version of $R$ whose rows, columns, and overall average have all been centered to zero. 

Using these quantities, we will consider the following $1-\alpha$ confidence upper bound on $\bar{\theta}$, whose coverage properties will be established in Section \ref{sec: CLT}:
\begin{equation}\label{eq: ideal CI}
 \CI^{\text{ideal}} = \hat{\theta} + z_{1-\alpha} \frac{\sqrt{\hat{\sigma}}}{L},
\end{equation}
where $z_{1-\alpha}$ is the $(1-\alpha)$ critical value of a standard normal. While $\CI^{\text{ideal}}$ cannot be evaluated in practice (as it requires knowledge of the unobserved vector $\theta$), it will be the basis for the eventual method.

\subsection{Asymptotic Coverage of $\CI^{\text{ideal}}$} \label{sec: CLT}

We assume a sequence of experiments whose components $X$, $Y$, $\{f_i\}$, $\{\eta_i\}$, $\theta$, $Z$ and statistic $T$ are given by \eqref{eq: X} - \eqref{eq: T} while $N\rightarrow \infty$, and which satisfies the following condition:
\begin{assumption}\label{as: CLT}
	There exist constants $B$, $D$, and $c > 0$ such that 
\begin{enumerate}
	\item The counterfactual vector $\theta$ is bounded:
	\[ 0 \leq \theta_i \leq B, \qquad i=1,\ldots,N \]
	\item Each neighborhood $\eta_i$ overlaps with a bounded number of other neighborhoods:
	\[\sum_{j:j\neq i} 1\{|\eta_i \cap \eta_j| > 0 \}  \leq D, \qquad i=1,\ldots,N, \]
	\item The variance of $T$ is lower bounded as a fraction of $N$:
	\[ \frac{1}{N} \operatorname{Var}(T) \geq c.\]
\end{enumerate}
\end{assumption}

Under Assumption \ref{as: CLT}, the following theorem establishes consistent estimation of $\operatorname{Var}(T)$ by $\hat{\sigma}$, and asymptotic normality of $\frac{T}{\sqrt{\operatorname{Var}(T)}}$ and $\frac{T}{\sqrt{\hat{\sigma}}}$:
\begin{theorem}\label{th: CLT}
	Let $p$, $p_{\min}$, $B$, $D$, and $c$ be fixed as $N \rightarrow \infty$, and let Assumption \ref{as: CLT} hold. It follows that 
\begin{enumerate}
	\item The estimate $\hat{\sigma}$ as given by \eqref{eq: hat sigma} converges to $\operatorname{Var}(T)$:
\begin{equation} \label{eq: hat sigma consistency}
\hat{\sigma} = \operatorname{Var}(T)\cdot (1 + o_P(1)).
\end{equation}
	\item The quantities $\frac{T}{\sqrt{\operatorname{Var}(T)}}$ and $\frac{T}{\sqrt{\operatorname{\hat{\sigma}}}}$ both converge in distribution to a standard normal random variable.
\end{enumerate}
\end{theorem}
An immediate corollary of Theorem \ref{th: CLT} is that the $(1-\alpha)$ confidence upper bound on $\bar{\theta}$ given by \eqref{eq: ideal CI} is asymptotically valid:
\begin{corollary} \label{cor: CLT}
	Under the conditions of Theorem \ref{th: CLT}, $\CI^{\text{ideal}}$ as given by \eqref{eq: ideal CI} is an asymptotically valid $(1-\alpha)$ confidence upper bound for $\bar{\theta}$:
	\[\mathbb{P}\left( \bar{\theta} \geq \CI^{\text{ideal}} \right) \leq \alpha + o(1).\]
\end{corollary}
Theorem \ref{th: CLT} and Corollary \ref{cor: CLT} are similar to \citep{aronow2017estimating} and are proven in the appendix.

\subsection{Upper Bounding $\CI^{\text{ideal}}$ Using Observations $Y$} \label{sec: upper bound}

Since $\theta$ is not observed, we cannot evaluate $\CI^{\text{ideal}}$ as given by \eqref{eq: ideal CI} in order to bound $\bar{\theta}$. However, since $\theta \leq Y$ by Assumption \ref{as: monotone}, conceptually we can upper bound the value of $\CI^{\text{ideal}}$ by finding its maximum over all values of $\theta$ allowed by the assumption:
\begin{equation}\label{eq: CI star}
\max_{\theta\in \mathbb{R}^N} \ \hat{\theta} + z_{1-\alpha} \frac{\sqrt{\hat{\sigma}}}{L} \quad \st \quad 0\leq \theta_i \leq Y_i,\ i=1,\ldots,N,
\end{equation}
where $\hat{\theta}$ and $\hat{\sigma}$ are functions of $\theta$ as given by \eqref{eq: theta hat} and \eqref{eq: hat sigma}. 

In general, the optimal choice of $\theta$ to maximize \eqref{eq: CI star} need not be $\theta = Y$; while setting $\theta = Y$ maximizes the point estimate $\hat{\theta}$, it need not maximize the variance estimate $\hat{\sigma}$. (A simulated example is given in Section \ref{sec: sims} using ``adversarial interference''.) As a result, it may be computationally difficult to solve \eqref{eq: CI star}. However, Theorem \ref{th: tilde CI} gives conditions under which a somewhat more conservative bound than \eqref{eq: CI star} can be computed. This bound, in which $\hat{\sigma}$ is replaced by an upper bound $\tilde{\sigma}$, is given by 
\begin{equation}\label{eq: tilde CI}
\max_{\theta\in \mathbb{R}^N} \ \hat{\theta} + z_{1-\alpha} \frac{\sqrt{\tilde{\sigma}}}{L}  \quad \st \quad 0\leq \theta_i \leq Y_i,\ i=1,\ldots,N,
\end{equation}
where $\tilde{\sigma}$ equals
\begin{equation}\label{eq: tilde sigma}
\tilde{\sigma} = Np(1-p)\left(\frac{1}{L}\sum_{i=1}^N (\theta_i - \hat{\theta})^2Z_i \right) + \sum_{i=1}^N \sum_{j=1}^N \theta_i \theta_j \frac{\max(\Escr_{ij},0)}{P_{ij}} Z_i Z_j.
\end{equation}
To see that $\tilde{\sigma}$ upper bounds $\hat{\sigma}$, observe that $\tilde{\sigma}$ replaces each term $\Escr_{ij}$ with $\max(0,\Escr_{ij})$, and  
\[\theta_i \theta_j \Escr_{ij} \leq \theta_i \theta_j \max(0, \Escr_{ij}), \qquad \text{ if } \theta_i,\theta_j \geq 0.\]
Theorem \ref{th: tilde CI} gives conditions under which $\theta = Y$ is known to maximimize \eqref{eq: tilde CI}:
\begin{theorem} \label{th: tilde CI}
	Let $\hat{\theta}_Y$ and $\tilde{\sigma}_Y$ be given by subsituting $\theta=Y$ into the expressions for $\hat{\theta}$ and $\tilde{\sigma}$ given by \eqref{eq: theta hat} and \eqref{eq: tilde sigma}, so that
\begin{align}
	\hat{\theta}_Y & = \frac{1}{L} \sum_{i=1}^N Y_i Z_i \label{eq: hat theta Y}\\
	\tilde{\sigma}_Y & = Np(1-p)\left(\frac{1}{L} \sum_{i=1}^N (Y_i - \hat{\theta}_Y)^2 Z_i \right)	 + \sum_{i=1}^N \sum_{j=1}^N Y_i Y_j \frac{\max(\Escr_{ij},0)}{P_{ij}}Z_iZ_j. \label{eq: tilde sigma Y}
\end{align}
	If it holds that
	\begin{equation}\label{eq: tilde CI condition}
	1 - z_{1-\alpha} \frac{\hat{\theta}_Y}{\sqrt{\tilde{\sigma}_Y}} \cdot \frac{2Np(1-p)}{L} \geq 0,
	\end{equation}
	then letting $\theta = Y$ maximizes \eqref{eq: tilde CI}, so that
	\begin{equation}\label{eq: Y CI}
	\CI^{\text{obs}} = \hat{\theta}_Y + z_{1-\alpha} \frac{\sqrt{\tilde{\sigma}_Y}}{L} 
	\end{equation}
	is an upper bound for $\CI^{\text{ideal}}$.
\end{theorem}

Theorem \ref{th: tilde CI} implies that $\CI^{\text{obs}}$ as given by \eqref{eq: Y CI} is an asymptotically conservative $1-\alpha$ confidence upper bound on $\bar{\theta}$, provided that condition \eqref{eq: tilde CI condition} is met and Assumptions \ref{as: monotone} and \ref{as: CLT} hold. The condition \eqref{eq: tilde CI condition} roughly requires $\sqrt{\tilde{\sigma}_Y}$, which estimates the standard deviation of $L(\hat{\theta} - \bar{\theta})$, to be large compared to $\hat{\theta}_Y$.
\subsection{Discussion} \label{sec: variants}

\paragraph{Variance Lower Bound} Assumption \ref{as: CLT} requires the lower bound 
\[ \frac{1}{N} \operatorname{Var}(T) \geq c, \]
to hold for some constant $c$. Similar to Condition 6 in \citep{aronow2017estimating}, Assumption \ref{as: CLT} is meant to rule out degenerate cases, such as when all values of $\theta$ are identical. In practice, one might wish to informally check this assumption by examining samples from $\theta$, or by estimating $\operatorname{Var}(T)$. However, $\theta$ is not observed in our setting, nor can $\operatorname{Var}(T)$ be estimated. Instead, we can only compute $\tilde{\sigma}_Y$, which does not lower bound $\operatorname{Var}(T)$ in any way. This raises the following question: even if $\tilde{\sigma}_Y$ is large, is it possible that $\operatorname{Var}(T)$ may be too small for the central limit theorem result of Theorem \ref{th: CLT} to hold?

To alleviate this concern, we give Theorem \ref{th: checking variance}, which states that if  $\tilde{\sigma}_Y/N$ exceeds $c$ by a constant factor, then even if $\operatorname{Var}(T)/N \geq c$ does not hold, the confidence bound of \eqref{eq: tilde CI} will  still be valid anyway:

\begin{theorem}\label{th: checking variance}
	Let $\hat{\sigma}_Y$ be given as \eqref{eq: tilde sigma Y}. If 
	\[\frac{1}{N} \tilde{\sigma}_Y \geq \frac{c}{z_{1-\alpha}^2 \alpha},\]
	and $\frac{1}{N}\operatorname{Var}(T) \leq c$, then $\bar{\theta}$ satisfies the upper bound of \eqref{eq: tilde CI}
	\[\bar{\theta} \leq \hat{\theta}_Y + z_{1-\alpha}\frac{\sqrt{\tilde{\sigma}_Y}}{L},\]
	  with probability $(1-\alpha)$.
\end{theorem}

As a result, if $\tilde{\sigma}_Y$ is large enough, then either $\operatorname{Var}(T)$ satisfies the lower bound required by Assumption \ref{as: CLT}, or Theorem \ref{th: checking variance} is satisfied so that \eqref{eq: tilde CI} is valid for finite $N$.

\paragraph{Estimation of outcomes under full control}

Let $\xi = (\xi_1,\ldots,\xi_N)$ denote the counterfactual outcome under zero treatment:
\[ \xi_i = f_i(0,\ldots,0), \qquad i=1,\ldots,N.\]
Analogous to Assumption \ref{as: monotone}, one might assume that
\[  Y_i \leq \xi_i \leq n_i, \qquad i=1,\ldots,N,\]
where $n_i$ is known and denotes the total enrollment at school $i$. A sufficient condition for this to hold is that
\[ \xi_i = \max_{x_{1:n}} f_i(x_1,\ldots,x_n), \qquad i=1,\ldots,N,\]
meaning that withholding treatment from all schools is assumed to give the worst outcomes. 

To find a $(1-\alpha)$ confidence lower bound on $\bar{\xi} = \frac{1}{N} \sum_{i=1}^N \xi_i$ under this assumption, it suffices to define $\tilde{Y} = (\tilde{Y}_1,\ldots,\tilde{Y}_N)$ and $\vartheta = (\vartheta_1,\ldots,\vartheta_N)$ by
\[ \tilde{Y}_i = n_i - Y_i \qquad \text{and} \qquad \vartheta_i = n_i - \xi_i, \qquad i=1,\ldots,N,\]
in which case it holds that  
\[ 0 \leq \vartheta_i \leq \tilde{Y}_i,\]
so that $\CI^{\text{obs}}$ \eqref{eq: Y CI} can be applied using observations $\tilde{Y}_{1:N}$ (instead of $Y_{1:N}$) to upper bound $\bar{\vartheta} = \frac{1}{N}\sum_{i=1}^N \vartheta_i$. Since $\bar{\xi} = \frac{1}{N} \sum_{i=1}^n (n_i - \vartheta_i)$, it follows that
\[ \frac{1}{N} \sum_{i=1}^N n_i - \CI^{\text{obs}}\]
is a $(1-\alpha)$ confidence lower bound on $\bar{\xi}$.

\paragraph{Nonuniform probabilities of effective treatment}

It may be of interest to consider exposure mappings such that
\[ \mathbb{P}(Z_i = 1) = \pi_i,\]
for some vector of nonuniform probabilities $\pi=(\pi_1,\ldots,\pi_N)$ with average value
\[ \bar{\pi} = \frac{1}{N} \sum_{i=1}^N \pi_i.\]
One might then consider a weighted point estimate of $\bar{\theta}$,
\[ \hat{\phi} = \frac{1}{L} \sum_{i=1}^N \frac{\bar{\pi}}{\pi_i} \theta_i Z_i \]
with $\hat{\sigma}$ estimating the variance of $L(\hat{\phi} - \bar{\theta})$ given by
\[	\hat{\sigma} = \frac{\sum_{i=1}^N \pi_i}{L} \sum_{i=1}^N (1 - \pi_i)(\phi_i - \hat{\phi})^2 Z_i + \sum_{i=1}^N \sum_{j=1}^N \phi_i \phi_j \frac{\Escr_{ij}}{P_{ij}} Z_i Z_j,\]
where $\phi_i = \frac{\bar{\pi}}{\pi_i} \theta_i$, and $R$ and $\Escr$ are given by
\begin{align*}
	R & = P - (\operatorname{diag}(\pi)(I - \operatorname{diag}(\pi)) - \pi \pi^T \\
	\Escr &= \left(I - \frac{1}{\bar{\pi}N} \obf\pi^T\right)^T R \left(I - \frac{1}{\bar{\pi}N} \obf\pi^T\right)^T. 
\end{align*}
For this setting, a central limit theorem analogous to Theorem \ref{th: CLT} can be shown to hold for the idealized confidence interval $\CI^{\text{ideal}}$. However, it is less clear how to modify Theorem \ref{th: tilde CI}; specifically, the appropriate condition analogous to \eqref{eq: tilde CI condition} does not seem clear, and is left as an open question.

\section{Simulation Study} \label{sec: sims}

By Theorem \ref{th: tilde CI}, $\CI^{\text{obs}}$ is an asymptotically valid confidence bound on $\bar{\theta}$ provided that the condition \eqref{eq: tilde CI condition} holds. We present here a simulation study to investigate how often the required condition \eqref{eq: tilde CI condition} holds, and whether in such cases the asymptotic result is a reasonable proxy for actual finite sample coverage.

\paragraph{Setup} In order to resemble the real-data setting of \citep{miguel2004worms}, the simulated units were chosen to have pairwise distances that were identical to the $N=49$ actual schools used in the study of \citep{miguel2004worms}. Treatments were assigned independently with Bernoulli probability $\rho=0.5$, and the observed outcomes $Y_i$ and counterfactuals $\theta_i$ were generated under four different scenarios:
\begin{enumerate}
	\item {\bf No effect, and no spatial clustering}: Treatment had no effect, so that $Y_i = \theta_i$ for all units $i$. The counterfactual outcomes $\theta_i$ were sampled without replacement from the observed infection counts of the schools in \citep{miguel2004worms}.
	\item {\bf No effect, but severe spatial clustering}: Treatment had no effect, so that $Y_i = \theta_i$ for all units $i$. The counterfactual outcomes $\theta_i$ were bi-modal based on geographic location, with $\theta_i=3$ for schools in the southern half of the map, and $\theta_i=15$ for schools in the northern half.
	\item {\bf Exposure model}: The counterfactual outcomes $\theta_i$ were sampled without replacement from the observed infection counts of the treated schools in \citep{miguel2004worms}. For schools that were directly treated and also had at least 2 of their nearest 5 neighbors directly treated, $Y_i = \theta_i$; otherwise, $Y_i$ was sampled uniformly from the untreated schools in \citep{miguel2004worms} whose observed outcomes exceeded $\theta_i$.
	\item {\bf Adversarial Interference}: $\theta_i$ was generated by sampling without replacement from a population where 2 units had $\theta_i=0$, 44 units had $\theta_i=10$, and 3 units had $\theta_i=20$. The observed outcomes were given by 
	\[Y_i = \begin{cases} 10 & \text{if $i$ treated and $\theta_i = 0$} \\ \theta_i & \text{otherwise}.\end{cases}\]
\end{enumerate}


\paragraph{Simulation Results}

When computing $\CI^{\text{obs}}$ as given by \eqref{eq: Y CI}, each neighborhood $\eta_i$ was chosen to equal $i$ plus $i$'s $d-1$ closest neighbors in geographic distance. The mappings $\{h_i\}$ were chosen to equal \eqref{eq: sim h}. The parameters $d_{\min}$ and $d$ were ranged over $d_{\min} \in \{2,\ldots,5\}$ and $d\in \{d_{\min},\ldots, 10\}$. In addition, the case $(d_{\min}, d)=(1,1)$, which ignores spatial information, was also explored. 

Table \ref{table: sim coverage} reports simulated coverage probabilities of $\CI^{\text{obs}}$ \eqref{eq: Y CI} for $\alpha = 1-0.95$, over those instances where the required condition \eqref{eq: tilde CI condition} was met.
Results are shown for all 4 generative models listed above, and for a subset of the explored values of $(d_{\min},d)$. Specifically, for $d_{\min} = 1,\ldots,5$, we report the smallest neighborhood size $d$ for which coverage exceeded 95\% in all scenarios, and for which coverage continued to exceed 95\% for all larger choices of $d$ as well; if no choice of $d$ meets this criterion, the $d$ which minimizes under-coverage is reported instead. 

For the first two scenarios when $d_{\min}$ was fixed, the coverage probabilities were observed to improve with increasing $d$, as this increased the number of units for which $Z_i=1$ (i.e., the ``sample size'') up to the set of all directly treated units. However, for the Exposure Model scenario, increasing $d$ too much led to conservative estimates. This is because the higher $d$ caused $Z_i$ to equal $1$ for treated units with large numbers of untreated neighbors; due to negative spillovers, these units had values of $Y_i$ which significantly exceeded their counterfactual $\theta_i$. This suggests the following intuition for choosing $d$ and $d_{\min}$:
\begin{enumerate}
	\item The set of units for which $Z_i=1$ should be large enough for the asymptotics of Theorem \ref{th: CLT} to be valid.
	\item However, the set of units for which $Z_i=1$ should not be too large, in order to exclude those treated units whose neighborhoods include large numbers of untreated units (in hopes of avoiding negative spillovers).
\end{enumerate}

Table \ref{table: good runs} shows the fraction of simulations for which the condition \eqref{eq: tilde CI condition} was met. This fraction is near 100\% in most cases, with the exception of Adversarial Interference when $(d_{\min}, d)=(1,1)$. For this scenario, we remark that coverage of $\CI^{\text{obs}}$ was only 77\% for the discarded instances where condition \eqref{eq: tilde CI condition} was not met, or 87\% coverage overall; in comparison, Table \ref{table: sim coverage} shows that coverage was $100\%$ when the condition was met. This reinforces that condition \eqref{eq: tilde CI condition} is necessary, or conversely that assuming $\theta_i = Y_i$ (essentially assuming SUTVA) can reduce coverage even when the interference is constrained to be nonnegative under Assumption \ref{as: monotone}. Our intuition why $(d_{\min},d) = (1,1)$ in particular is sensitive to the condition \eqref{eq: tilde CI condition} is that $R=0$ in this case, meaning that $\tilde{\sigma}=\hat{\sigma}$ and is no longer conservative. Evidently the conservativeness of $\tilde{\sigma}$ is beneficial in the Adversarial Interference setting, where the interference may reduce the empirical variance.

\begin{center} 
\begin{table}[h!] 
\begin{tabular}{ l l | r r r r  } 
 &  &  No effect, & No effect, & Exposure &  \\ 
$d_{\min}$ & $d$  &  no clustering & severe clustering &  model &  Adversarial \\
\hline
1 & 1 &     0.92 &   0.96 &   0.99 &   1.00 \\
2 & 3 &     0.95 &   0.96 &   0.98 &   1.00 \\
3 & 6 &     0.96 &   0.95 &   0.96 &   1.00 \\
4 & 10 &     0.95 &   0.95 &   0.99 &   1.00 \\
5 & 10 &     0.96 &   0.90 &   0.99 &   0.99		
\end{tabular} 
\caption{Simulated coverage probabilities of $\CI^{\text{obs}}$ under scenarios described in Section \ref{sec: sims}, for $d_{\min}=1,\ldots,5$, and for each $d_{\min}$ the smallest choice of $d$ which coverage exceeded 95\% in all scenarios (and also for all larger choices of $d$); if no choice of $d$ satisfies this requirement, the $d$ which minimizes under-coverage is reported instead. Each table value based on 1000 simulations.}
\label{table: sim coverage}
\end{table}
\end{center}

\begin{center}
\begin{table}[h!] 
\begin{tabular}{ l l | r r r r  }
 &  &  No effect, & No effect, & Exposure &  \\ 
$d_{\min}$ & $d$  &  no clustering & severe clustering &  model &  Adversarial \\
\hline
1 & 1 &    1.00  & 1.00  &  1.00  &  0.46 \\
2 & 3 &    1.00  &  1.00 &   1.00 &   0.97 \\
3 & 6 &    0.99  & 0.99  &  1.00  &  0.98 \\
4 & 10 &    1.00  &  1.00 &   1.00 &   0.98 \\
5 & 10 &    0.98  & 0.96  &  0.98  &  0.96
\end{tabular} 
\caption{Fraction of simulations for which condition \eqref{eq: tilde CI condition} was met, so that $\CI^{\text{obs}}$ could be applied to generate a valid confidence set. We remark that for the Adversarial scenario with $(d_{\min},d) = (1,1)$, coverage was reduced to $87\%$ if condition \eqref{eq: tilde CI condition} was ignored. Each table value based on 1000 simulations.}
\label{table: good runs}
\end{table}
\end{center}

\section{Data Analysis Example} \label{sec: real data}

As a pedagogical example, we apply $\CI^{\text{obs}}$ \eqref{eq: Y CI} to the data example of \citep{miguel2004worms}, using the neighborhoods $\{\eta_i\}$ and mappings $\{h_i\}$ tested in the simulations. Guided by the central limit theorem behavior under the simulations, we considered $(d_{\min}, d) \in \{(2,3), (3,6),(4,10)\}$; as shown in Table \ref{table: sim coverage}, these were the parameter values for which simulated coverage exceeded 95\% for all scenarios. Condition \eqref{eq: tilde CI condition} was met for all of these choices, so that in each case, $\CI^{\text{obs}}$ induced an asymptotically valid confidence upper bound on $\bar{\theta}$. 


Table \ref{table: real data} reports the estimated confidence intervals induced by $\CI^{\text{obs}}$, for the chosen values of $(d_{\min},d)$. For comparison, the confidence interval for $(d_{\min}, d) = (1,1)$, which ignores spatial information, is also shown. We see that usage of spatial information results in better (i.e., less conservative) estimates, as the intervals for $(2,3)$ and $(3,6)$ are smaller than that of $(1,1)$.

As we examined three different choices of the spatial parameters $(d_{\min},d)$, a Bonferoni-corrected confidence set may be appropriate. For $(d_{\min},d) = (3,6)$, the Bonferoni-corrected value set using $\CI^{\text{obs}}$ equaled $[0,\, 6.1]$, which still improves over the non-Bonferoni-corrected estimate of $[0,\, 7.0]$ attained without spatial information. However, validity of the Bonferoni-corrected $\CI^{\text{obs}}$ requires asymptotic normality to be valid not at the $0.95$ quantile, but rather at the more extreme quantile of $1 - .05/3 = 98.3\%$. In simulations, we found that actual coverage at this level was $97\%$, indicating that the coverage of the Bonferoni-corrected CI may only be approximate. 

We emphasize that our analysis is primarily pedagogical: as the number of schools is not particularly large ($N=49$), it is perhaps unsurprising that the asymptotic confidence intervals began to lose coverage after Bonferoni correction. For small datasets that may commonly arise in non-internet applications, it may be preferable to follow the traditional approach of specifying a single exposure model, in which case $\CI^{\text{obs}}$ yields estimates that are robust to misspecification. For larger datasets, it may be profitable to go further and explore multiple choices of exposure model (with correction for multiple testing); we leave the tast of efficiently scanning for the best choice of neighborhood as a subject for future work.

\begin{table}[h!] 
\begin{center}
\begin{tabular}{l|r }
	& Confidence Set \\
	$(d_{\min},d)$ & $[0,\, \CI^{\text{obs}}]$ \\
	\hline
	(1,1) & [0,\, 7.0] \\
	(2,3) & [0,\, 6.6] \\
	(3,6) & [0,\, 5.6] \\
	(4,10) & [0,\, 7.8]
\end{tabular} 
\end{center}
\caption{Estimates of $\bar{\theta}$ for the data of \citep{miguel2004worms} under varying neighborhoods parameterized by $(d_{\min},d)$.}
\label{table: real data}

\end{table}


\section{What can be learned with no assumptions on interference?} \label{sec: nonmonotone}

In this section we propose a new estimand, for which valid (though possibly conservative) estimation is possible under completely arbitrary interference. Naturally, this estimand will be more limited than existing ones which require structural assumptions. We envision that this estimand may be useful as the initial result of an analysis. For example, in order to clarify the understanding of a contentious research question, it might be advantageous to present a ``layered analysis'' in which the initial findings are cautious; these findings would be limited in the scope of their implication, but would ideally require no leaps of faith in the causal reasoning. 

In previous work, this type of limited analysis could be done by rejecting a null model (but without a corresponding confidence interval on the effect size), as proposed by  \citep{athey2017exact}, or by using rank-based estimands proposed by \citep{rosenbaum2007interference}. Our estimand may be more interpretable or descriptive than these previous approaches, giving quantitative bounds on the ``contrast attributable to treatment'', which will be closely related to attributable treatment effects \citep{rosenbaum2001effects}. 

\subsection{Definition of Estimand}

As before, let $X=(X_1,\ldots,X_N)$ and $Y = (Y_1,\ldots,Y_N)$ denote the randomized treatments and observed outcomes, under general interference so that each $Y_i$ is an idiosyncratic function of all $N$ treatment assignments
\[ Y_i = f_i(X_1,\ldots,X_N),\qquad i=1,\ldots,N,\]
and let $N_0$ and $N_1$ denote the number of units with treatments zero and one respectively:
\[N_0 = \sum_{i=1}^N 1\{X_i=0\}\qquad \text{and} \qquad N_1 = \sum_{i=1}^N 1\{X_i=1\}.\]
Let $\xi = (\xi_1,\ldots,\xi_N)$ denote the counterfactual outcome under full control,
\[\xi_i = f_i(0,\ldots,0),\]
with average value $\bar{\xi} = \frac{1}{N} \sum_{i=1}^N \xi_i$. Unlike before, we will make no assumptions on the values of $\xi$.

Let $\Delta_Y$ denote the observed contrast between the treated and untreated units,
\[\Delta_Y = \frac{1}{N_1} \sum_{i:X_i=1} Y_i - \frac{1}{N_0} \sum_{i:X_i=0} Y_i,\]
and let $\Delta_\xi$ denote the contrast between the same units, but under the counterfactual of full control:
\[\Delta_\xi = \frac{1}{N_1} \sum_{i:X_i=1} \xi_i - \frac{1}{N_0} \sum_{i:X_i=0} \xi_i.\]
We remark that the quantity $\Delta_\xi$ is the difference between two sample averages, and hence will concentrate at zero under mild assumptions, such as when $\xi$ is bounded and $X$ is generated randomly and independently of $\xi$.

Let $\tau^{\text{CAT}}$ denote the {\it constrast attributable to treatment}, defined as the difference between $\Delta_Y$ and $\Delta_{\xi}$:
\[ \tau^{\text{CAT}} = \Delta_Y - \Delta_\xi.\]
If $\tau^{\text{CAT}} > 0$, the causal implication is that the treatment changed the value of the contrast between treated and control, shifting it in favor of the treated population. This means that the treatment caused the treated units to have higher outcomes than the control units. However, it does not specify whether the contrast was caused by an increase in the outcomes of the treated units, or a decrease in the outcomes of the control units.

To generalize $\tau^{\text{CAT}}$, let $Z = (Z_1,\ldots,Z_N)$ denote an effective treatment indicator, based on a set of neighborhoods $\{\eta_i\}$ and exposure mappings $\{h_i\}$ as before, so that
\[ Z_i = h_i(\{Z_j: j \in \eta_i\}),\]
and recall that $L = \sum_{i=1}^N Z_i$. Let $\Delta_{Y,Z}$ denote the observed contrast between the units in effective treatment and control
\[\Delta_{Y,Z} = \frac{1}{L} \sum_{i:Z_i=1} Y_i - \frac{1}{N - L} \sum_{i:Z_i=0} Y_i.\]
Let $\Delta_{\xi,Z}$ denote the contrast between the same units, but under the counterfactual of full control:
\[\Delta_{\xi,Z} = \frac{1}{L} \sum_{i:Z_i=1} \xi_i - \frac{1}{N-L} \sum_{i:Z_i=0} \xi_i.\]
Let $\tau^{\text{ZCAT}}$ denote the {\it $Z$-induced contrast attributable to treatment}, given by
\[ \tau^{\text{ZCAT}} = \Delta_{Y,Z} - \Delta_{\xi,Z}.\]
The quantity $\tau^{\text{ZCAT}}$ has a similar causal interpretation as  $\tau^{\text{CAT}}$, except that the effect of treatment is measured on the contrast between a different division of units.

\subsection{Estimation of $\tau^{\text{CAT}}$ and $\tau^{Z\text{CAT}}$}

Theorem \ref{th: CAT} gives conditions under which $\Delta_Y$ consistently estimates $\tau^{\text{CAT}}$, and gives (potentially conservative) asymptotic confidence intervals, both one- and two- sided:

\begin{theorem} \label{th: CAT}
	Let $X = (X_1,\ldots,X_N)$ denote binary treatments assigned by sampling without replacement, and let $Y = (Y_1,\ldots,Y_N)$ denote binary outcomes generated under arbitrary interference. Then with probability converging to at least $1-\alpha$, it holds that
	\begin{align}
		\tau^{\text{CAT}} & \geq \Delta_Y - \frac{z_{1-\alpha}}{2}\sqrt{\frac{N}{N_0N_1}}, \label{eq: CAT one-sided}
	\end{align}
	as well as
	\begin{align}
		|\tau^{\text{CAT}} - \Delta_Y| &\leq \frac{z_{1-\alpha/2}}{2}\sqrt{\frac{N}{N_0N_1}}, \label{eq: CAT two-sided}
	\end{align}
	implying that $\Delta_Y = \tau^{\text{CAT}} + o_P(1)$ as $\min(N_1,N_0) \rightarrow \infty$.
\end{theorem}

Theorem \ref{th: ZCAT} gives conditions under which $\Delta_{Y,Z}$ consistently estimates $\tau^{\text{ZCAT}}$, as well as confidence intervals. Unlike Theorem \ref{th: CAT} which assumes sampling without replacement, in Theorem \ref{th: ZCAT} we assume $X$ is generated by Bernoulli randomization. To state the theorem, we recall quantities $P$ and $T$:
\begin{align*}
	P_{ij} &= \mathbb{P}(Z_iZ_j =1),\qquad i,j=1,\ldots,N \\
	T(Z;\xi) &= \sum_{i=1}^N (\xi_i - \bar{\xi})Z_i.
\end{align*}
\begin{theorem}\label{th: ZCAT}
Let $X = (X_1,\ldots,X_N)$ denote binary treatments assigned independently by Bernoulli randomization, and let $Y = (Y_1,\ldots,Y_N)$ denote binary outcomes generated under arbitrary interference. Let the following hold
\begin{align*}
\operatorname{Var}(T) &\geq cN  \\
\sum_{j:j\neq i} 1\{|\eta_i \cap \eta_j| > 0\} & \leq D,  \qquad i=1,\ldots,N,\\
P_{ii} &= p,\qquad i=1,\ldots,N,
\end{align*}
for constants $c>0$, $D$, and $p$ which are fixed as $N \rightarrow \infty$. Then it holds with probability converging to at least $1-\alpha$ that
\begin{align*}
\tau^{Z\text{CAT}} & \geq \Delta_{Y,z} - \frac{z_{1-\alpha}}{2p(1-p)}\cdot\sqrt{\frac{\lambda_1}{N}} 
\end{align*}
and
\begin{align*}
\left| \Delta_Y - \tau^{Z\text{CAT}}\right| & \leq \frac{z_{1-\alpha/2}}{2p(1-p)}\cdot\sqrt{\frac{\lambda_1}{N}},
\end{align*}
where $\lambda_1$ is the largest eigenvalue of $\left(I - \frac{\obf\obf^T}{N}\right)P\left(I - \frac{\obf \obf^T}{N}\right)$. 
\end{theorem}
\subsection{Data Example (Facebook Voting Experiment) }

The paper \citep{bond201261} describes a voting experiment, in which Facebook users were encouraged through an online advertisement to self-report that they had voted by clicking on an ``I voted'' button. For the units randomly assigned to treatment, the advertisement also contained the profile pictures of up to six Facebook friends who had already self-reported. This means that for each viewer, the content of the advertisement depended on the actions of previous viewers, possibly leading to interference. 

Table \ref{table: facebook} gives the reported counts for the experiment, rounded for display. For these values, Theorem \ref{th: CAT} implies that with confidence converging to at least 95\%, $\tau^{\text{CAT}}$ lies within the interval $[2.06\%, \, 2.26\%]$, or equivalently that $\Delta_\xi \in [-0.1\%, 0.1\%]$. This means that with high confidence the profile pictures caused the discrepancy in rates of self-reported voting to shift from roughly zero to approximately 2\% in favor of the treated population. This estimate holds with no assumptions on the interference, meaning that peer effects could be heterogeneous, positive or negative, long range/dense/global, and could include factors such as unobserved offline interactions between individuals, global influence of actors such as mass media or news aggregators, negative backlash in oversaturated areas, or untreated units becoming discouraged upon discovering the nature of the treatment.

\begin{table}[h!] 
\begin{center}
\begin{tabular}{l|r |r }
	& Control & Treated \\
	\hline
	Total Count & 611K & 60M \\
	Clicked ``I voted'' & 109K & 12M 
\end{tabular} 
\end{center}
\caption{Counts  (rounded for display) from the Facebook election experiment of \citep{bond201261}.}
\label{table: facebook}
\end{table}

\subsection{Discussion}

The estimand $\tau^{\text{CAT}}$ is significantly more limited than traditional estimands. It does not yield any information on whether the treatments improved outcomes for treated units or worsened outcomes for control units. Nor does it yield any information on the counterfactuals of full treatment or full control (i.e., $\bar{\theta}$ or $\bar{\xi}$), nor on the expected value of $\Delta_Y$ or $N^{-1}\sum_i Y_i$. 

In the presence of interference, it seems reasonable to believe that estimation of such quantities will generally require subjective judgement. For example, to estimate the outcome if all units were treated, given an experiment where only a small fraction are exposed to an advertisement, one would have to discern whether units might become fatigued by oversaturation of the advertisement, leading to diminishing effects or even changes in the sign of the effect. To show concentration of the observed average to its expectation, one would have to judge the extent to which the average outcome could be shifted by a single actor with global influence.

We are {\it not} claiming that these subjective judgements should never be made, or that such quantities should not be estimated. Rather, in poorly-understood settings where these modeling decisions may lack consensus, we hope to add clarity by defining a ``baseline'' or ``fallback'' that can be inferred from the randomization alone. This is in contrast to the analyst's best estimate or ``best guess'', which may utilize assumptions -- such as an exposure model or partial interference -- requiring the full breadth of the analyst's subjective expertise . In practice, it may be for some settings that presenting both types of estimates -- those that minimize assumptions and those that utilize the full (and necessarily subjective for contentious topics) expertise of the analyst -- may be the most informative for the broadest audience.

 \bibliographystyle{apalike}
 \bibliography{bibfile}

%
%
%
%
%
%

%
%

\pagebreak

\appendix

\section*{Supplemental Materials}

\section{Proof of Theorem \ref{th: CLT} and Corollary \ref{cor: CLT}}

The proof of Theorem \ref{th: CLT} will use Lemmas \ref{le: CLT helper variance identity}, \ref{le: CLT helper mcdiarmid}, and \ref{le: hat sigma term 2}, as well as Theorem \ref{th: local dependence} from \cite[Th 2.7]{chen2004normal}, which gives a central limit theorem under local dependence. Lemmas \ref{le: CLT helper variance identity}, \ref{le: CLT helper mcdiarmid}, and \ref{le: hat sigma term 2} are proven in Section \ref{sec: CLT helper lemmas}.

\begin{lemma} \label{le: CLT helper variance identity}
	It holds that
\begin{align} \label{eq: var T helper 1}
	\operatorname{Var}(T) &= N p(1-p)  \left(\frac{1}{N} \sum_{i=1}^N (\theta_i - \bar{\theta})^2\right) + \theta^T \Escr \theta 
\end{align}
	
\end{lemma}

\begin{lemma} \label{le: CLT helper mcdiarmid}
	Let the conditions of Theorem \ref{th: CLT} hold. Then
	\begin{align} 
		\frac{1}{L} \sum_{i=1}^N (\theta_i - \hat{\theta})^2Z_i &= \frac{1}{N} \sum_{i=1}^N (\theta_i - \bar{\theta})^2 + o_P(1). \label{eq: hat sigma term 1}
	\end{align}
\end{lemma}

\begin{lemma} \label{le: hat sigma term 2}
		Let the conditions of Theorem \ref{th: CLT} hold. Then
		\begin{align}
					\sum_{i=1}^N \sum_{j=1}^N \theta_i \theta_j \frac{\Escr_{ij}}{P_{ij}} Z_i Z_j &= \theta^T \Escr \theta + o_P(N).	\label{eq: hat sigma term 2}
		\end{align}
\end{lemma}

\begin{theorem}[\citep{chen2004normal}, Theorem 2.7] \label{th: local dependence}
	Let $\{w_i: i \in \mathcal{V}\}$ be random variables indexed by the vertices of a dependency graph with maximal degree $D$. Put $W = \sum_{i\in \mathcal{V}} w_i$. Assume that $\mathbb{E}W^2=1$, $\mathbb{E}w_i=0$, and $\mathbb{E}|w_i|^s \leq C^s$ for $i \in \mathcal{V}$, $s \in (2,3]$, and for some $C >0$. Then
	\[\sup_z |\mathbb{P}(W \leq z) -\Phi(z) | \leq 75 D^{5(s-1)}|\mathcal{V}|C^s,\]
	where $\Phi(x)$ is the CDF of a standard normal.	
\end{theorem}

\begin{proof}[Proof of Theorem \ref{th: CLT}]
	The two claims are proven as follows
	\begin{enumerate}
	\item To show \eqref{eq: hat sigma consistency}, observe that
	\begin{align*}
		\hat{\sigma} & = N p(1-p)\left(\frac{1}{L}\sum_{i=1}^N (\theta_i - \hat{\theta})^2 Z_i\right)  + \sum_{i=1}^N \sum_{j=1}^N \theta_i \theta_j \frac{\Escr_{ij}}{P_{ij}} Z_i Z_j \\
		& = N p(1-p)\left(\frac{1}{N} \sum_{i=1}^N (\theta_i - \bar{\theta})^2  \right) +  \theta^T \Escr \theta + o_P(N)	\\
		& = \operatorname{Var}(T) + o_P(N) \\
		& = \operatorname{Var}(T)(1 + o_P(1)),
	\end{align*}
	where the first equation is \eqref{eq: hat sigma}, the second follows from \eqref{eq: hat sigma term 1} and \eqref{eq: hat sigma term 2}, the third is \eqref{eq: var T helper 1}, and the fourth uses the inequality $\operatorname{Var}(T) \geq cN$ from Assumption \ref{as: CLT}.

\item To show that $\frac{T}{\sqrt{\operatorname{Var}(T)}}$ and $\frac{T}{\sqrt{\hat{\sigma}}}$ is asymptotically normal, let $w_i$ be given by
\[ w_i = \frac{(\bar{\theta} - \theta_i)Z_i}{\sqrt{\operatorname{Var}(T)}},\]
so that $W  = \sum_{i=1}^N w_i$ equals $\frac{T}{\sqrt{\operatorname{Var}(T)}}$. By Assumption \ref{as: CLT}, each neighbor $\eta_i$ has overlap with at most $D$ other neighborhoods. As a result, the variables $\{w_i\}$ form a dependency graph with maximal degree $D$. Assumption \ref{as: CLT} also enforces that $|\theta_i| \leq B$ and $\operatorname{Var}(T) \geq cN$, which together imply that
\[ |w_i | \leq \frac{2B}{\sqrt{cN}}.\]
As a result, Theorem \ref{th: local dependence}, it follows for $s \in (2,3]$ that 
	\[\sup_z |\mathbb{P}(W \leq z) -\Phi(z) | \leq 75 D^{5(s-1)}N\left(\frac{2B}{\sqrt{cN}}\right)^s.\]	
The left hand side goes to zero for $s > 2$ when $D, B$ and $c$ are fixed as $N \rightarrow \infty$, proving convergence in distribution of $\frac{T}{\sqrt{\operatorname{Var}(T)}}$ to a standard normal.

	To show that $\frac{T}{\sqrt{\hat{\sigma}}}$ is asymptotically normal, use asymptotic normality of $\frac{T}{\sqrt{\operatorname{Var}(T)}}$, and substitute $\hat{\sigma} = \operatorname{Var}(T)(1 + o_P(1))$.
\end{enumerate}
\end{proof}

\begin{proof}[Proof of Corollary \ref{cor: CLT}]
	By Theorem \ref{th: CLT}, it holds that $\frac{T}{\sqrt{\hat{\sigma}}}$ is asymptotically normal. Thus with probability converging to $1-\alpha$, it holds that
	\[\frac{T}{\sqrt{\hat{\sigma}}} \cdot \frac{L}{L} \leq z_{1-\alpha}.\]
	Substituting $T = \sum_{i=1}^N (\bar{\theta} - \theta_i)Z_i$ and rearranging terms yields that with probability converging to $1-\alpha$,
	\begin{align*}
	\bar{\theta} & \leq \frac{1}{L} \sum_{i=1}^N \theta_i Z_i + z_{1-\alpha} \frac{\sqrt{\hat{\sigma}}}{L} \\
	& = \hat{\theta} + z_{1-\alpha} \frac{\sqrt{\hat{\sigma}}}{L}, 
	\end{align*}
	proving the corollary.
\end{proof}

\subsection{Proof of Lemmas \ref{le: CLT helper variance identity}, \ref{le: CLT helper mcdiarmid}, and \ref{le: hat sigma term 2}} \label{sec: CLT helper lemmas}

\begin{proof}[Proof of Lemma \ref{le: CLT helper variance identity}]
	The identity \eqref{eq: var T helper 1} holds by
	\begin{align*}
		\operatorname{Var}(T) &= (\theta - \bar{\theta})^T P (\theta - \bar{\theta}) \\
		&= (\theta - \bar{\theta})^T\left( p(1-p)I + p^2 11^T + R\right)(\theta - \bar{\theta}) \\
		&= (\theta - \bar{\theta})^T(p(1-p)I)(\theta - \bar{\theta}) + (\theta - \bar{\theta})^TR(\theta - \bar{\theta}) \\
		&= p(1-p)\sum_{i=1}^N(\theta - \bar{\theta})^2 + \theta^T\left(I - \frac{11^T}{N}\right)^TR\left(I - \frac{11^T}{N}\right)\theta \\
		&= p(1-p)\sum_{i=1}^N(\theta - \bar{\theta})^2 + \theta^T\Escr\theta, 
	\end{align*}
	where the first equality is \eqref{eq: var T}; the second holds by \eqref{eq: R}; the third holds because $(\theta - \bar{\theta})^T11^T(\theta - \bar{\theta}) = 0$; the fourth because $(\theta - \bar{\theta}) = (I - 11^T/N)\theta$; and the fifth by \eqref{eq: Escr}.
		
\end{proof}

The proof of Lemmas \ref{le: CLT helper mcdiarmid} and \ref{le: hat sigma term 2} will use Theorem \ref{th: azuma}, which is a well-known concentration inequality

\begin{theorem}[Azuma-Hoeffding, Method of Bounded Differences]\label{th: azuma}
	Let $f:\mathbb{R}^N \mapsto \mathbb{R}$ satisfy the bounded difference property with constants $m_1,\ldots,m_N$, meaning that
	\[ | f(x_1,\ldots,x_N) - f(x_1',\ldots,x_N') | \leq m_i,\]
	whenever $x_j = x_j'$ for all $j\neq i$. Let $X_1,\ldots,X_N$ be independent random variables. It holds that 
	\[ \mathbb{P}\left( |f(X_1,\ldots,X_N) - \mathbb{E}f(X_1,\ldots,X_N)| > t\right) \leq 2\exp\left(-\frac{2t^2}{M}\right),\]
	where $M = \sum_{i=1}^N m_i^2$.
\end{theorem}

\begin{proof}[Proof of Lemma \ref{le: CLT helper mcdiarmid}]	
	To prove  \eqref{eq: hat sigma term 1},  we will require the following intermediate results:
	\begin{align}
		L &= Np(1+o_P(1)) \label{eq: mcdiarmid 1} \\
		\hat{\theta} &= \bar{\theta} + o_P(1) \label{eq: mcdiarmid 2} \\
		\frac{1}{Np} \sum_{i=1}^N (\theta_i - \bar{\theta})^2Z_i &= \frac{1}{N} \sum_{i=1}^N (\theta_i - \bar{\theta})^2 + o_P(1), \label{eq: mcdiarmid 3}
	\end{align}
	which we prove as follows:
	\begin{enumerate}
		\item To prove \eqref{eq: mcdiarmid 1},	let $f_L(X_1,\ldots,X_N) = \sum_{i=1}^N Z_i$. Since $Z_i = \prod_{j \in \eta_i} X_j$, and each neighborhood $\eta_i$ overlaps with at most $D$ other neighborhoods, it can be seen that $f_L$ has the bounded difference property with constants $m_i=D$ for all $i=1,\ldots,N$. As a result, Theorem \ref{th: azuma} with $M = ND^2$ implies that
	\[ \mathbb{P}\left( |f_L - \mathbb{E}f_L| > \epsilon\cdot N\right) \leq 2\exp\left(-\frac{2\epsilon^2 N}{D^2}\right),\]
	implying that $f_L = \mathbb{E}f_L + o_P(N)$. Since $L = f_L$ and $\mathbb{E}f_L = Np$, this implies that 
	\begin{equation*}
	L = Np + o_P(N).
	\end{equation*}
	Since $p$ is fixed under Assumption \ref{as: CLT}, this proves \eqref{eq: mcdiarmid 1}.
	
	\item To prove \eqref{eq: mcdiarmid 2}, let $f_{\hat{\theta}}$ be given by
	\[ f_{\hat{\theta}}(X_1,\ldots,X_N) = \frac{1}{Np} \sum_{i=1}^N \theta_i Z_i.\]
	Since $L = Np(1+o_P(1))$, it holds that
	\begin{equation} \label{eq: helper hat theta}
	\hat{\theta} = f_{\hat{\theta}}(X_1,\ldots,X_N)\cdot (1+o_P(1)).
	\end{equation}
	By similar reasoning as before, $f_{\hat{\theta}}$ satisfies the bounded difference property with constants $m_i = \frac{DB}{Np}$. Theorem \ref{th: azuma} (with $M = \frac{D^2B^2}{Np^2}$) implies that
	\[ \mathbb{P}\left( |f_{\hat{\theta}} - \mathbb{E}f_{\hat{\theta}}| \geq \epsilon\right) \leq 2\exp\left(-\frac{2\epsilon^2Np^2}{D^2B^2}\right).\]
	Since $\mathbb{E}f_{\hat{\theta}} = \bar{\theta}$, this implies that
	\[ f_{\hat{\theta}}(X_1,\ldots,X_N) = \bar{\theta} + o_P(1),\]
	which in turn implies by \eqref{eq: helper hat theta} that $\hat{\theta} = \bar{\theta} + o_P(1)$.

	\item The proof of \eqref{eq: mcdiarmid 3} is identical to that of \eqref{eq: mcdiarmid 2}, using $f(X_1,\ldots,X_N) = (Np)^{-1}\sum_{i=1}^N (\theta_i - \bar{\theta})^2Z_i$.
	\end{enumerate}
	Using these intermediate results, we now derive \eqref{eq: hat sigma term 1}:
	\begin{align*}
		\frac{1}{L} \sum_{i=1}^N (\theta_i - \hat{\theta})^2 Z_i & = \frac{1}{Np(1+o_P(1))} \sum_{i=1}^N (\theta_i - \bar{\theta} + o_P(1))^2 Z_i \\
		& = (1 + o_P(1)) \cdot \frac{1}{Np} \sum_{i=1}^N (\theta_i - \bar{\theta})^2 Z_i  + o_P(1)\\
		& = (1 + o_P(1)) \cdot \frac{1}{N}\sum_{i=1}^N (\theta_i + \bar{\theta})^2 + o_P(1),
	\end{align*}
	where the first inequality follows by substituting \eqref{eq: mcdiarmid 1} and \eqref{eq: mcdiarmid 2}; the second by algebraic manipulation; and the third by \eqref{eq: mcdiarmid 3}.
\end{proof}

\begin{proof}[Proof of Lemma \ref{le: hat sigma term 2}]
	Let $r_i$ and $c_j$ denote the row and column means of the matrix $R$, and let $\mu$ denote the overall mean:
	\[ r_i = \frac{1}{N}\sum_{j=1}^N R_{ij}, \qquad c_j = \frac{1}{N} \sum_{i=1}^N R_{ij}, \qquad \text{and} \qquad \mu = \frac{1}{N^2} \sum_{i=1}^N \sum_{j=1}^N R_{ij}.\]
	Then $\Escr$ satisfies 
	\[ \Escr_{ij} = R_{ij} - r_j - c_j + \mu.\]
	Let $f:\{0,1\}^N\mapsto \mathbb{R}$ denote the function
	\begin{align*}
	f(X_1,\ldots,X_N) &= \sum_{i=1}^N \sum_{j=1}^N \theta_i \theta_j \frac{\Escr_{ij}}{P_{ij}}Z_iZ_j,
	\end{align*}
	and let $f_1$ and $f_2$ denote the functions
	\begin{align*}
		f_1(X_1,\ldots,X_N) & = \sum_{i=1}^N \sum_{j=1}^N \theta_i \theta_j \frac{R_{ij}}{P_{ij}}Z_iZ_j \\
		f_2(X_1,\ldots,X_N) & = \sum_{i=1}^N \sum_{j=1}^N \theta_i \theta_j \frac{r_i + c_j - \mu}{P_{ij}}Z_iZ_j.
	\end{align*}
	Then $f$ satisfies
	\[ f(X_1,\ldots,X_N) = f_1(X_1,\ldots,X_N) + h_2(X_1,\ldots,X_N).\]
	Let $C_p$ denote the constant
	\[C_p = \max\left(\frac{p-p^2}{p_{\min}}, \frac{p^2-p_{\min}}{p_{\min}}\right).\]
	As an intermediate result, we show now that $h_1$ and $h_2$ satisfy the bounded difference property, so that the following holds:
	\begin{align}
		|f_1(x_1,\ldots,x_N) - f_1(x_1',\ldots,x_N')| &\leq 2B^2D^2C_p \label{eq: BD 1}\\
		|f_2(x_1,\ldots,x_N) - f_2(x_1',\ldots,x_N')| &\leq 4B^2D^2C_p, \label{eq: BD 2}
	\end{align}
	whenever $x_i = x_i'$ for all $i=1,\ldots,N$ except for a single coordinate. We show \eqref{eq: BD 1} and \eqref{eq: BD 2} below:
	
	\begin{enumerate}
		\item Given $x_1,\ldots,x_N$ and $x_1',\ldots,x_N'$ which differ in a single element, let $\{Z_i\}$ and $\{Z_i'\}$ be induced in the obvious way
		\[ Z_i = \prod_{j \in \eta_i} x_j \qquad \text{and} \qquad Z_i' = \prod_{j \in \eta_i} x_j', \qquad i=1,\ldots,N.\]
		Let $S \subset [N]$ denote the indices for which $Z_i \neq Z_i'$. It follows that
		\begin{align} 
		\nonumber	 & \hskip-1cm |f_1(x_1,\ldots,x_N) - f_1(x_1',\ldots,x_N')|  \\
		\nonumber	 & \hskip-0cm \leq   \sum_{i \in S} \sum_{j=1}^N \left| \theta_i \theta_j \frac{R_{ij}}{P_{ij}}  (Z_iZ_j - Z_i'Z_j')\right| \\
		\nonumber & \hskip.5cm  +  \sum_{i \notin S} \sum_{j \in S} \left| \theta_i \theta_j \frac{R_{ij}}{P_{ij}}(Z_iZ_j - Z_i'Z_j') \right| \\
		\nonumber & \hskip.5cm + \sum_{i \notin S} \sum_{j \notin S} \left| \theta_i \theta_j \frac{R_{ij}}{P_{ij}} (Z_iZ_j - Z_i' Z_j')\right| \\
		\nonumber	 & \hskip-0cm =  \sum_{i \in S} \sum_{j=1}^N \left| \theta_i \theta_j \frac{R_{ij}}{P_{ij}} \right| + \sum_{j\in S} \sum_{i \notin S} \left| \theta_i \theta_j \frac{R_{ij}}{P_{ij}} \right| \\
					& \hskip-0cm \leq  \sum_{i \in S} \sum_{j=1}^N B^2 \left|\frac{R_{ij}}{P_{ij}} \right| + \sum_{j\in S} \sum_{i=1}^N B^2 \left| \frac{R_{ij}}{P_{ij}} \right|, \label{eq: BD 1 helper}
		\end{align}
		where the first inequality follows by simple algebra; the second follows because $|Z_iZ_j - Z_i'Z_j'| \leq 1$, with $Z_iZ_j = Z_i'Z_j'$ if $i,j \notin S$; and the third follows because $\theta_i \leq B$.
		
		To bound the right hand side of \eqref{eq: BD 1 helper}, recall the definition of $R$
		\[R = P - p(1-p)I - p^211^T,\]
		and recall that $p_{\min} \leq P_{ij} \leq p$, which implies that $p_{\min} - p^2 \leq R_{ij} \leq p - p^2$ and hence that $|R_{ij}/P_{ij}| \leq C_p$. Also recall that $P_{ij} = p^2$ if $\eta_i$ and $\eta_j$ are disjoint. It follows that each row/column of $P$ has at most $D$ entries not equal to $p^2$, and hence that each row/column of $R$ has at most $D$ nonzero entries. As a result, the right hand side of \eqref{eq: BD 1 helper} can be bounded by
		\[  |f_1(x_1,\ldots,x_N) - f_1(x_1',\ldots,x_N')| \leq 2B^2D^2 C_p,\]
		where we have used the fact that $|S| \leq D$ if $\{x_i\}$ and $\{x_i'\}$ differ in only a single element. This proves \eqref{eq: BD 1}.
		
		\item Analogous to before, it holds that 
		\begin{align} 
		\nonumber	 & \hskip-1cm |f_2(x_1,\ldots,x_N) - f_2(x_1',\ldots,x_N')|  \\
		\nonumber	 & \leq \sum_{i \in S} \sum_{j=1}^N \left| \theta_i \theta_j \frac{r_i + c_j - \mu}{P_{ij}} \right| +  \sum_{i \notin S} \sum_{j \in S} \left| \theta_i \theta_j \frac{r_i + c_j - \mu}{P_{ij}} \right| \\
		\nonumber	 & \hskip1cm {} + \sum_{i \notin S} \sum_{j \notin S} \left| \theta_i \theta_j \frac{r_i + c_j - \mu}{P_{ij}} (Z_iZ_j - Z_i'Z_j')\right| \\
				 	& \leq \sum_{i \in S} \sum_{j=1}^N B^2 \left| \frac{r_i + c_j - \mu}{P_{ij}} \right| +  
						\sum_{j \in S} \sum_{i=1}^N B^2 \left| \frac{r_i + c_j - \mu}{P_{ij}} \right|.  \label{eq: BD 2 helper}
		\end{align}		
		Since each row/column of $R$ has at most D nonzero entries which also satisfy $p_{\min} - p^2 \leq R_{ij} \leq p - p^2$, it follows that
		\[ \frac{D}{N}(p_{\min} - p^2) \leq r_i \leq \frac{D}{N}p(1-p) \qquad \text{and} \qquad \frac{D}{N}(p_{\min}-p^2) \leq c_j \leq \frac{D}{N}p(1-p),\]
		as well as
		\[ \frac{D}{N}(p_{\min}-p^2) \leq \mu \leq \frac{D}{N} p(1-p).\]
		As a result, it follows that 
		\[ \left|\frac{r_i + c_j - \mu}{P_{ij}}\right| \leq 2\frac{D}{N}C_p,\]
		so that the right hand side of \eqref{eq: BD 2 helper} can be bounded by
		\begin{align*}
			  |f_2(x_1,\ldots,x_N) - f_2(x_1',\ldots,x_N')| &\leq 4B^2DN\cdot\frac{D}{N}C_p \\
			  = 4B^2D^2C_p,
		\end{align*}
		proving \eqref{eq: BD 2}.
	\end{enumerate}
	Using \eqref{eq: BD 1} and \eqref{eq: BD 2}, we can apply Theorem \ref{th: azuma} to show that 
	\[ \mathbb{P}\left(|f - \mathbb{E}f| > \epsilon N\right) \leq \exp\left(-\frac{2\epsilon^2N}{36B^4D^4C_p^2}\right),\]
	which implies that $|f - \mathbb{E}f| = o_P(N)$. Since $\mathbb{E}f$ equals
	\[ \mathbb{E}f(X_1,\ldots,X_N) = \sum_{i=1}^N \theta_i \theta_j \Escr_{ij},\]
	this establishes \eqref{eq: hat sigma term 2}, proving the lemma.
	
\end{proof}

\section{Proof of Theorems \ref{th: tilde CI} and \ref{th: checking variance}}

\begin{proof}[Proof of Theorem \ref{th: tilde CI}]
Proof by contradiction. Let $f$ denote the objective function of \eqref{eq: tilde CI}
\[ f(\theta) = \hat{\theta} + z_{1-\alpha} \frac{\sqrt{\tilde{\sigma}}}{L},\]
let \eqref{eq: tilde CI condition} hold, and suppose there exists $\theta^* \in \mathbb{R}^L$ which is a feasible solution to \eqref{eq: tilde CI} whose objective value satisfies $f(\theta^*) > f(Y)$. Let $\theta(t)$ for $t \in [0,1]$ be given by
\[ \theta(t) = Y + t(\theta^* - Y),\]
so that $\theta(0) = Y$ and $\theta(1) = \theta^*$. Let $\hat{\theta}(t)$ and $\hat{\sigma}(t)$ denote the empirical mean and standard deviation of $\theta(t)$, as given by \eqref{eq: theta hat} and \eqref{eq: tilde sigma}.

Observe that $f$ is increasing in $\hat{\theta}$ and $\tilde{\sigma}$. Since $\hat{\theta}(t) \leq \hat{\theta}_Y$ for all $t \in [0,1]$, it must hold that $\tilde{\sigma}(1) > \tilde{\sigma}_Y$ in order for $f(\theta(1)) > f(Y)$. Since $\tilde{\sigma}(1) > \tilde{\sigma}_Y$ and $\tilde{\sigma}(0) = \tilde{\sigma}_Y$, it follows from continuity of $\tilde{\sigma}(t)$ that exists $t_0 < 1$ which satisfies
\[ t_0 = \sup \{t \in [0,1]: \tilde{\sigma}(t) \leq \tilde{\sigma}_Y\}\]
must exist and satisfy $t_0 < 1$. Since it also holds that
\[ f(\theta(1)) = f(\theta(t_0)) + \int_{t_0}^1 \nabla f(\theta(t))^T (\theta^* - Y)\, dt,\]
it follows that if we show that $f(\theta(t_0)) \leq f(Y)$ and $\nabla f(\theta(t))^T (\theta^* - Y) \leq 0$ for all $t \in [t_0,1]$, then we have shown that $f(\theta(1)) \leq f(Y)$, contradicting our claim that $f(\theta(1)) > f(Y)$. We prove these below:
\begin{enumerate}
	\item Since $\hat{\theta}(t_0) \leq \hat{\theta}_Y$ and $\hat{\sigma}(t_0) \leq \hat{\sigma}_Y$, it follows that $f(\theta(t_0)) \leq f(Y)$. 
	\item To show  $\nabla f(\theta(t))^T (\theta^* - Y) \leq 0$ for all $t \in [t_0,1]$, it suffices to show that $\theta^* - Y \leq 0$ and $\nabla f(\theta(t)) \geq 0$  for all $t \in [t_0,1]$ (both elementwise):
	\begin{enumerate}
		\item The condition $\theta^* - Y \leq 0$ is implied by the constraint set of \eqref{eq: tilde CI}
		\item To show that $\nabla f(\theta(t)) \geq 0$ elementwise, let $g(\theta)$ denote
	\[ g(\theta) = \sum_{i=1}^N \sum_{j=1}^N \theta_i \theta_j \frac{\max(0,\Escr_{ij})}{P_{ij}} Z_i Z_j,\]
	and lower bound the derivative $\frac{df}{d \theta_i}$ by
\begin{align} \label{eq: derivative bound}
\nonumber	\frac{df}{d\theta_i} &= \frac{d\hat{\theta}}{d\theta_i} + \frac{z_{1-\alpha}}{2L}\cdot \frac{1}{\sqrt{\tilde{\sigma}}} \frac{d\tilde{\sigma}}{d\theta_i} \\
\nonumber &= \frac{Z_i}{L} + \frac{z_{1-\alpha}}{2L}\cdot\frac{1}{\sqrt{\tilde{\sigma}}}\left(Np(1-p)\left(\frac{2\theta_iZ_i}{L} - 2\hat{\theta}\frac{Z_i}{L}\right) + \frac{dg}{d\theta_i}\right)\\
	& \geq \frac{Z_i}{L} - \frac{z_{1-\alpha}}{L} \frac{\hat{\theta}}{\sqrt{\tilde{\sigma}}} \cdot\left(\frac{Np(1-p)}{L}\right)Z_i 
\end{align}
where the second line holds because $\frac{d\hat{\theta}}{d\theta_i}$ and $\frac{d\tilde{\sigma}}{d\theta_i}$ are given by
\[ \frac{d\hat{\theta}}{d\theta_i} = \frac{Z_i}{L} \qquad \textrm{and} \qquad \frac{d\tilde{\sigma}}{d\theta_i} = Np(1-p)\left(\frac{2}{L} \theta_iZ_i - \frac{2}{L}\hat{\theta}Z_i\right) + \frac{dg}{d\theta_i},\]
the third line holds by lower bounding $\theta_i \geq 0$ and $\frac{dg}{d\theta_i} \geq 0$, and the last inequality holds by \eqref{eq: tilde CI condition}.

Combining \eqref{eq: tilde CI condition} and \eqref{eq: derivative bound} implies that $\frac{df}{d\theta_i} \geq 0$ for all $(\hat{\theta}, \tilde{\sigma}$ satisfying $\hat{\theta} \leq \hat{\theta}_Y$ and $\tilde{\sigma} \geq \tilde{\sigma}_Y$. Since these conditions both hold for $\{\theta(t): t \geq t_0\}$, it follows that $\nabla f(\theta(t)) \geq 0$ for all $t \in [t_0,1]$. 
\end{enumerate}

\end{enumerate}
This contradicts the existence of $\theta^* \neq Y$ such that $f(\theta^*) > f(Y)$, proving the theorem.

\end{proof}

\begin{proof}[Proof of Theorem \ref{th: checking variance}]
	By Chebychev's inequality, it holds that
	\[\mathbb{P}\left(|T| \geq \frac{\sqrt{\operatorname{Var}(T)}}{\sqrt{\alpha}} \right) \leq \alpha,\]
	or equivalently after dividing by $L$,
	\[\mathbb{P}\left(\frac{|T|}{L} \geq \frac{1}{\sqrt{\alpha}} \cdot \frac{\sqrt{\operatorname{Var}(T)}}{L} \right) \leq \alpha.\]
	Substituting $T = \sum_{i=1}^N (\bar{\theta}-\theta_i)$ and rearranging terms yields that with probability at least $1-\alpha$,
	\[\bar{\theta} \leq \hat{\theta} + \frac{1}{\sqrt{\alpha}} \frac{\sqrt{\operatorname{Var}(T)}}{L}.\]
	Substituting $\operatorname{Var}(T) \leq \hat{\sigma}_Y z_{1-\alpha}^2 \alpha$ (which is given by the statement of Theorem \ref{th: checking variance}) yields the desired result. 
\end{proof}

\section{Proof of Theorems \ref{th: CAT} and \ref{th: ZCAT}}

\begin{proof}[Proof of Theorem \ref{th: CAT}]
	By algebraic manipulation, $\Delta_\xi$ can be written as
	\begin{align*}
		\Delta_\xi & = \sum_{X_i=1} \frac{\xi_i}{N_1}  -  \sum_{X_i=0} \frac{\xi_i}{N_0}  \\
		& = \sum_{X_i=1} \xi_i\left(\frac{1}{N_1} + \frac{1}{N_0}\right) - \sum_{i=1}^N \frac{\xi_i}{N_0} \\
		& = \sum_{X_i=1} \xi_i \frac{N}{N_1N_0} - \frac{N}{N_0} \sum_{i=1}^N \frac{\xi_i}{N} \\
		& = \frac{N}{N_0}\cdot \frac{1}{N_1} \sum_{X_i=1} (\xi_i - \bar{\xi}),
	\end{align*}
	where $\bar{\xi} = \frac{1}{N}\sum_{i=1}^N \xi_i$.
	
	Let $\hat{\psi}$ denote the quantity
	\[ \hat{\psi} = \frac{1}{N_1} \sum_{i:X_i=1} (\xi_i - \bar{\xi}).\]
	By properties of sampling without replacement \citep{thompson2002sampling}, it holds that
	\[ \operatorname{Var}(\hat{\psi}) = \sigma^2 \frac{1}{N_1}\cdot\frac{N-N_1}{N},\]
	where $\sigma^2$ equals the variance of the entries of $\xi$,
	\[ \sigma^2 = \frac{1}{N_1-1}\sum_{i=1}^N (\xi_i - \bar{\xi})^2.\]
	Since $\xi$ is a binary vector, it holds that $\sigma^2 = \bar{\xi}(1-\bar{\xi}) \leq 1/4$.
	
	By a finite population central limit theorem \citep{thompson2002sampling}, it holds that $\frac{\hat{\psi}}{\sqrt{\operatorname{Var}(\hat{\psi})}}$ converges to a standard normal distribution as $\min(N_1,N_0) \rightarrow \infty$. This implies that
	\[ \mathbb{P}\left( \frac{\hat{\psi}}{\sqrt{\operatorname{Var}(\hat{\psi})}} \geq z_{1-\alpha}\right) = \alpha + o(1).\]
	Multiplying both sides of the inequality inside the probability expression by $\frac{N}{N_0}\sqrt{\operatorname{Var}(\hat{\psi})}$ yields
	\[ \mathbb{P}\left( \Delta_\xi \geq z_{1-\alpha} \frac{N}{N_0} \sqrt{\operatorname{Var}(\hat{\psi})}  \right) = \alpha + o(1),\]
	where we have used the identity that $\Delta_\xi = \frac{N}{N_1}\hat{\psi}$. Since $\sigma^2 \leq 1/4$, it follows that $\operatorname{Var}(\hat{\psi}) \leq \frac{1}{4}\frac{N-N_1}{N_1N}$, and by algebraic manipulation it holds that
	\[ \mathbb{P}\left( \Delta_\xi \geq \frac{z_{1-\alpha}}{2} \sqrt{\frac{N}{N_0N_1}} \right) \leq \alpha + o(1).\]
	Since $\tau^{\text{CAT}} = \Delta_Y - \Delta_\xi$, this implies the one-sided interval \eqref{eq: CAT one-sided}. The two-sided interval \eqref{eq: CAT two-sided} follows by parallel arguments.
\end{proof}

\begin{proof}[Proof of Theorem \ref{th: ZCAT}]
	Following the proof of Theorem \ref{th: CAT}, $\Delta_{\xi,Z}$ can be written as
		\begin{align*}
			\Delta_{\xi,Z} & = \frac{N}{(N-L)L}\cdot \sum_{Z_i=1} (\xi_i - \bar{\xi}) \\
			& = \frac{N}{(N-L)L} \cdot T(Z;\xi).
		\end{align*}
	Abbreviating $T \equiv T(Z;\xi)$, we observe that $\operatorname{Var}(T)$ equals
	\begin{align*}
	\operatorname{Var}(T) &= (\xi - \bar{\xi})^TP(\xi - \bar{\xi}) \\
	&= (\xi - \bar{\xi})^T\left(I - \frac{\obf\obf^T}{N}\right) P \left(I - \frac{\obf \obf^T}{N}\right) (\xi-\bar{\xi}),
	\end{align*}
	where we have used the identity $(\xi-\bar{\xi}) = (\xi - \bar{\xi})^T\left(I - \frac{\obf\obf^T}{N}\right)$, which holds because $(\xi-\bar{\xi})^T\obf = 0$. It thus follows that $\operatorname{Var}(T)$ can be upper bounded by
	\begin{align*}
	\operatorname{Var}(T) &\leq \|\xi - \bar{\xi}\|^2 \lambda_1 \\
	& \leq \frac{N}{4} \lambda_1
	\end{align*}
	where $\lambda_1$ is the largest eigenvalue of $\left(I - \frac{\obf\obf^T}{N}\right) P \left(I - \frac{\obf \obf^T}{N}\right)$, and the second inequality uses the fact that $\xi$ is a binary vector, so the variance of its entries is bounded by $1/4$.
	
	Under the conditions of the theorem, Part 2 of the proof of Theorem \ref{th: CLT} implies that $\frac{T}{\sqrt{\operatorname{Var}(T)}}$ converges to a standard normal distribution. By algebraic manipulation, it follows that
	\[ \mathbb{P}\left( \frac{N}{L(N-L)} T(Z;\xi) \geq z_{1-\alpha} \frac{N}{L(N-L)} \sqrt{\operatorname{Var}(T)}\right) = \alpha + o(1).\]
	Substituting $\Delta_{\xi,Z} = \frac{N}{L(N-L)}T(Z;\xi)$ and the upper bound $\operatorname{Var}(T) \leq N\lambda_1/4$ yields (after some algebraic manipulation) that
	\[ \mathbb{P}\left( \Delta_{\xi,Z} \geq \frac{z_{1-\alpha}}{2p(1-p)}\sqrt{\frac{\lambda_1}{N}} \right) \leq \alpha + o(1),\]
	where we have used the result $L = Np(1+o_P(1))$ which is \eqref{eq: mcdiarmid 1} in the proof of Lemma \ref{le: CLT helper mcdiarmid}. As $\tau^{Z\text{CAT}} = \Delta_Y - \Delta_{\xi}$, this proves the one-sided confidence interval. The two-sided interval is proven by parallel arguments.
\end{proof}

\end{document}